\documentclass[letterpaper, 10 pt, conference]{ieeeconf}
\IEEEoverridecommandlockouts                              % This command is only needed if
\usepackage{times}
\usepackage{multicol}
\usepackage[bookmarks=true]{hyperref}
\usepackage{cite}
\usepackage{amsmath,amssymb,amsfonts}
\usepackage{graphicx}
\usepackage{tikz}
\usetikzlibrary{arrows.meta,positioning,matrix}
\usepackage{pgfplots}
\pgfplotsset{compat=1.18}
\usepackage{paralist}

\usepackage{graphicx} % Required for inserting images
\usepackage[export]{adjustbox}
\usepackage{float}
\usepackage{xcolor}
\usepackage{url}

\newtheorem{definition}{Definition}

\newtheorem{proposition}{Proposition}

\newtheorem{remark}{Remark}
\newtheorem{example}{Example}
\newtheorem{assumption}{Assumption}

\newcommand{\reals}{\mathbb{R}}
\newcommand{\nnreals}{\mathbb{R}_{\geq 0}}
\newcommand{\naturals}{\mathbb{N}}

\newcommand{\rulebook}{\mathcal{R}}
\newcommand{\ruleset}{R}
\newcommand{\rbrule}{r}
\newcommand{\realizations}{\Xi}
\newcommand{\preorder}{\lesssim}

\newcommand{\outcomes}{\Omega}
\newcommand{\algebra}{\Sigma}
\newcommand{\probability}{\mathit{Pr}}
\newcommand{\expectation}{\mathbb{E}}
\newcommand{\VaR}{\text{VaR}}
\newcommand{\CVaR}{\text{CVaR}}
\newcommand{\outcome}{\omega}
\newcommand{\risk}{\rho}
\newcommand{\mfunctions}{\mathcal{F}}
\newcommand{\interaction}{\mathbf{E}}

\newcommand{\threshold}{\gamma}
\newcommand{\rrisk}{\risk_{\rbrule}}

\newcommand{\rriski}[1]{\risk_{\rbrule_{#1}}}
\newcommand{\rthreshold}{\threshold_{\rbrule}}

\newcommand{\rthresholdi}[1]{\threshold_{\rbrule_{#1}}}

\newcommand{\rrrule}{\rbrule_{\text{risk}}}
\newcommand{\rarulebook}{\rulebook_{\text{risk}}}
\newcommand{\raruleset}{R_{\text{risk}}}

\newcommand{\rtrajectories}{\mathcal{T}}
\newcommand{\etrajectories}{\mathcal{E}}
\newcommand{\rtraj}{\tau}
\newcommand{\etraj}{\xi}

\newcommand{\trajectories}{\mathcal{T}}

\newcommand{\rbleq}{\lesssim_{\rulebook}}
\newcommand{\raleq}{\lesssim_{\rarulebook}}
\newcommand{\rblt}{<_{\rulebook}}
\newcommand{\ralt}{<_{\rarulebook}}

\title{\LARGE \bf
  Risk-Aware Rulebooks for Multi-Objective Trajectory Evaluation under Uncertainty
}

\author{Tichakorn Wongpiromsarn$^{1}$% <-this % stops a space
  \thanks{*This work was supported in part by NSF under Grant
    CNS-2141153.}% <-this % stops a space
  \thanks{$^{1}$Tichakorn Wongpiromsarn is with the Department of Computer Science,
    Iowa State University, Ames, IA 50010
    {\tt\small nok@iastate.edu}}%
}

\begin{document}

\maketitle
\thispagestyle{empty}
\pagestyle{empty}

\begin{abstract}
  We present a risk-aware formalism for evaluating system trajectories in the presence of
  uncertain interactions between the system and its environment.
  The proposed formalism supports reasoning under uncertainty and systematically handles complex
  relationships among requirements and objectives, including hierarchical priorities and
  non-comparability.
  % Specifically, we extend the rulebook formalism to account for uncertainty arising from the
  % interaction between the controlled system and its environment, enabling the assessment and
  % comparison of alternative trajectories prior to execution.
  Rather than treating the environment as exogenous noise,
  we explicitly model how each system trajectory influences the environment and evaluate
  trajectories under the resulting distribution of environment responses.
  We prove that the formalism induces a preorder on the set of system trajectories, ensuring
  consistency and preventing cyclic preferences.
  Finally, we illustrate the approach with an autonomous driving example
  that demonstrates how the formalism enhances explainability by clarifying the rationale behind
  trajectory selection.

  % In this context, safety requirements can be treated as objectives,
  % where high-priority requirements, such as those that affect human safety,
  % take precedence, while others may lack a clear hierarchy or comparability.
  % Compounding this challenge,
  % decision-making in these systems must account for uncertainties stemming from
  % imperfect perception and unstructured and unpredictable environments.
  % To address these challenges, this paper focuses on providing a formal specification language
  % capable of handling conflicting objectives that are only partially comparable
  % and enable reasoning under uncertainty.
  % These rules often differ in importance.
  % For example, avoiding collisions is universally regarded as more critical than staying within lane
  % boundaries.
  % However, the relative priority between certain rules, such as staying within lane boundaries and
  % giving sufficient clearance to other vehicles can be less straightforward and often left to the
  % discretion of system developers.
\end{abstract}

\IEEEpeerreviewmaketitle

\section{Introduction}
Safety-critical autonomous systems, such as autonomous vehicles, must satisfy multiple safety
requirements, including avoiding collisions, adhering to local driving norms,
and obeying traffic laws.
At the same time, these systems need to optimize various objectives, such as
minimizing travel time and maximizing passenger comfort.
However, these requirements and objectives often conflict,
and in some cases, it may not even be possible to satisfy all safety requirements simultaneously
\cite{Topcu:2020:Assured}.
For example, on a highway with a 65 mph speed limit, where most humans drive at 80 mph,
an autonomous vehicle faces a dilemma between adhering to the speed limit or adapting to the flow of
traffic, with both choices potentially increasing the risk of accidents.
This complexity necessitates a framework that supports formally justifiable trade-offs,
where certain requirements may be violated if necessary.
While requirements such as those that affect human safety are more important than others,
other requirements and objectives may not be comparable.
Moreover, decision-making in these systems needs to take into account uncertainties stemming from
imperfect perception and unstructured and unpredictable environments.

Variants of temporal logics have been employed as formal specification languages to
precisely define the requirements of safety-critical systems
\cite{Baier:2008:Principles,Dokhanchi:2018:Evaluating,Arechiga:2019:Specifying}.
Recent extensions, such as Metric Temporal Logic \cite{Fainekos:2009:Robustness},
robust Linear-time Temporal Logic \cite{Anevlavis:2021:Being}, and Signal Temporal Logic
(STL) \cite{Maler:2004:Monitoring,Donze:2010:Robust,Raman:2014:Model,Raman:2015:Reactive} include
robust semantics and allow quantitative assessment of the robustness of satisfaction.
Various robustness metrics have been proposed
\cite{Akazaki:2015:Time,Lindemann:2017:Robust,Haghighi:2019:Control,Mehdipour:2019:Arithmetic}.
Probabilistic extensions, including Probabilistic Signal Temporal Logic
\cite{Sadigh:2016:Safe,Tiger:2020:Incremental} and Chance Constrained Temporal Logic \cite{Jha:2018:Safe},
allow reasoning about correctness in the presence of uncertainty.
More recent formulations have further incorporated risk measures into temporal-logic-based
planning and control
\cite{Lindemann:2021:STL,Lindemann:2022:Reactive,Akella:2025:Risk,Lindemann:2023:Risk,Safaoui:2022:Risk},
enabling robustness-risk quantification, reactive synthesis, and risk-bounded control of stochastic
systems.
% \citet{Lindemann:2021:STL} introduces the notion of robustness risk for STL that quantifies the
% risk of not satisfying an STL specification robustly and allows for general risk measures and
% \cite{Lindemann:2022:Reactive} generalizes STL to include both uncontrollable propositions and
% risk predicates to enable both reactive and risk-aware planning.
These existing approaches, however, combine all the requirements into a single formula, failing to
capture unequal importance and complex relationships among individual requirements.

To address unequal importance among requirements, prior works have imposed
a total preorder over specifications, each expressed in linear temporal logic
\cite{Tumova:2013:ACC,Tumova:2013:LCS,Castro:2013:CDC,Tumova:2014:Maximally,Tumova:2016:Least,Vasile:2017:Minimum,Wongpiromsarn:2021:Minimum}.
The violation of each requirement is quantified by the amount of time the system spends violating
it.
The rulebook formalism \cite{Censi:2019:Liability,Wongpiromsarn:2026:Formal} provides a unified representation of
requirements and objectives by abstracting both as \emph{rules}.
Each rule is associated with a cost function that quantifies its degree of unsatisfaction
(e.g., the severity of violation for safety requirements or deviation from the intended goal for
performance objectives).
This formalism extends prior approaches in two key ways.
First, it allows each rule to be equipped with its own independently defined violation metric
(e.g., kinetic energy transfer for collision avoidance),
instead of using the same violation metric for all specifications, as is typical in
temporal-logic-based approaches.
Second, instead of imposing a total preorder over rules, it permits a general preorder,
which allows both hierarchical priorities and non-comparable relationships.
A preorder relation encodes the relative importance of the rules
and determines how trade-offs are resolved.
As a result, the rulebook formalism generalizes existing temporal logic-based and optimization-based
planning and control, providing a unifying framework across robotics, formal methods, control
theory, and operations research \cite{Wongpiromsarn:2026:Formal}.
% a systematic method for evaluating system behavior in
% the presence of complex and potentially non-comparable objectives.
However, the current rulebook formalism is limited to retrospective evaluation,
where the system's behavior can only be assessed after execution, once the
environmental outcomes (e.g., the movements of other vehicles or pedestrians), are
fully realized.

% each safety requirement or objective is represented by a cost function that quantifies its
% violation and serves as a criterion to be minimized.
% A rulebook consists of a set of rules, each of which represents a distinct objective and is
% associated with a cost function. This cost function quantifies the extent of the rule's
% unsatisfaction: for safety requirements aimed at minimizing violations, the cost function measures
% the severity or degree of violation, whereas for performance objectives, it evaluates how far the
% system deviates from achieving its intended goal.
% that quantifies its violation.
% Additionally, a preorder relation is used to capture the relative importance of these rules,
% enabling the specification of both hierarchical and non-comparable relationships among objectives.
% As a result, this framework offers a systematic approach to evaluate a robot's behavior in the
% presence of complex relationships among the objectives.
% However, the current rulebook formalism is limited to retrospective evaluation,
% where the robot's behavior can only be assessed after its actions have been executed
% and the actual environmental outcomes, such as the movements of other vehicles or pedestrians, are fully realized.

The main contribution of this paper is the extension of
the rulebook formalism to enable the evaluation and comparison of candidate
system trajectories under environmental uncertainty during the planning stage.
This provides a systematic approach for comparing alternative system trajectories,
whether generated by rule-based or machine-learning-based methods, before execution.
Specifically, our contributions are as follows.
\vspace{-1mm}
\begin{itemize}
\item We develop a risk-aware rulebook formalism that explicitly captures system-environment
  interaction under uncertainty by allowing system trajectories to influence the distribution of environmental
  behaviors.
  Each pair of rule and system trajectory then induces a random variable
  representing the resulting rule violation, whose risk is quantified using a chosen risk measure.
  This enables a unified treatment of expectation-based, worst-case, and tail-risk criteria to be
  handled within a single preordered rule structure and
  provides a systematic approach for comparing alternative system trajectories,
  whether generated by rule-based or machine-learning-based methods, before execution.
\item Our construction ensures that the resulting risk-aware rulebook is a
  rulebook on the set of system trajectories.
  As a result, it inherits the key properties of classical rulebooks
  \cite{Censi:2019:Liability,Wongpiromsarn:2026:Formal}.
  In particular, it induces a preorder over system trajectories. This prevents cyclic preferences
  (e.g., trajectory A is better than B, B is better than C, but C is better than A) and
  ensures a well-defined notion of ``optimal'' trajectory.
  Moreover, existing rulebook-based control synthesis algorithms
  \cite{Wongpiromsarn:2026:Formal} can be applied without modification.
  Finally, this makes it natural for a system designer to start
  from a classical rulebook (specified, e.g., by a regulator and used for retrospective evaluation with
  fully realized environmental outcomes)
  and augment each rule with an appropriate risk measure to handle uncertainty during planning.
\item We establish a connection between safety and optimality under this formalism.
  In addition, we show that optimal trajectories satisfy principled tradeoff properties:
  if an alternative trajectory strictly improves upon an optimal trajectory under one rule,
  then it must incur a compensating deterioration under another rule that is not lower in priority.
  In this sense, optimal trajectories are rational as their selection is supported by explicit and
  structurally consistent tradeoffs across competing objectives under uncertainty.
  By making these tradeoffs explicit, our formalism provides a basis for explainability: preferences
  between trajectories can be traced to well-defined rule-level comparisons governed by the rule
  preorder, clarifying how uncertainty and partially prioritized objectives jointly determine the
  selected trajectory.
\end{itemize}

% This enables the system to make informed, risk-aware decisions while ensuring that the evaluation
% process aligns with partially prioritized objectives even in uncertain environments.
% Additionally, our formalism accounts for the interaction between the system and its environment,
% modeling how a system trajectory can influence the behavior of the surrounding entities.
% We provide a complete proof that the proposed risk-aware rulebook formalism induces a preorder on
% the set of system trajectories.
% This ensures consistency by preventing cyclic preferences (e.g., trajectory A is better than B, B is
% better than C, but C is better than A); thus, ensuring that the notion of ``optimal'' trajectory is
% well-defined.
% We also establish a connection between ``safe'' and ``optimal'' trajectories under this
% formalism.
% As a result, our formalism provides explainability by making the rationale behind the system's
% decisions transparent, clarifying how conflicting objectives and uncertainties are managed to
% identify an optimal trajectory.

\section{Preliminaries}
\label{sec:prelim}
Throughout the paper, we let $\reals$, $\nnreals$, and $\naturals$ denote the set of real,
non-negative real, and natural numbers, respectively.
The expectation of a random variable $f$ is denoted by $\expectation[f]$.

\subsection{Rulebooks}
This section provides an overview of a formalism called ``rulebooks'' introduced in
\cite{Censi:2019:Liability}.
At the core of this framework is the concept of ordering rules based on their relative priority,
which defines preference relations over a set of possible
\emph{realizations}, denoted by $\realizations$.
% Each rule represents a violation metric on the possible outcomes.
We now provide precise definitions of the foundational elements of this formalism.

\begin{definition}
  \label{def:rule}
  A  \emph{rule} is a function $\rbrule : \realizations \to \mathbb{R}_{\geq 0}$ that measures the
  degree of violation of its argument.
\end{definition}

If $\rbrule(x) < \rbrule(y)$, then the realization $y$ violates the rule $\rbrule$ to a greater
extent than does $x$.
In particular, $\rbrule(x) = 0$ indicates that a realization $x$ is fully compliant with the rule.
Note that the definition of the violation metric might be analytical, ``from first principles'', or
be the result of a learning process.
Additionally, while we refer to this function as a ``rule'', it is not limited to a regulatory
constraint but can also represent performance-related criteria
like comfort, efficiency, or user preferences that are desirable
but not strictly required.

The rulebook formalism uses the notion of a preorder \cite{Schechter:1997:Handbook,Eklund:1990:Generalized}
to rigorously define how rules can be ordered.

\begin{definition}
  \label{def:preorder}
  A \emph{preorder} on a set $S$ is a binary relation $\preorder$ that is
  reflexive ($s \preorder s$ for all $s \in S$), and
  transitive ($s_{1} \preorder s_{2}$ and $s_{2} \preorder s_{3}$ imply $s_{1} \preorder s_{3}$ for
  all $s_{1}, s_{2}, s_{3} \in S$).
  Any preorder gives rise to a strict partial order $<$ defined by $s_{1} < s_{2}$ whenever
  $s_{1} \preorder s_{2}$ but $s_{2} \not \preorder s_{1}$.
\end{definition}

Note that a preorder is not necessarily antisymmetric.
In other words, it is possible for both $s_{1} \preorder s_{2}$ and $s_{2} \preorder s_{1}$ to hold
for some $s_{1} \not= s_{2}$.
As a result, preorder can be seen as a generalization of partial order.
% Specifically, every partial order is a preorder but a preorder that is not antisymmetric is not a partial order.

% Given a preorder $\preorder$, we can define an equivalence relation $\sim$ on $S$ such that
% $s_{1} \sim s_{2}$ if and only if $s_{1} \preorder s_{2}$ and $s_{2} \preorder s_{1}$.
% It can be shown that the resulting relation is reflexive, transitive, and symmetric.
% Using this equivalence relation, a preorder on $S$ can be viewed as a partial order on
% the set $\equivset{S}$ of all equivalence classes of $\sim$.
% the quotient set of the equivalence $\equivset{S}$, which is the set of all equivalence classes of $\sim$.
% a preorder on $S$ can be viewed as an equivalence relation on $S$,
% together with a partial order on the set of equivalence class.

\begin{definition}
  \label{def:rulebook}
  A rulebook is a tuple $\rulebook = \langle \ruleset, \preorder \rangle$,
  where $\ruleset$ is the set of rules and $\preorder$ is a preorder on $\ruleset$,
  which specifies the relative importance among the rules.
\end{definition}

Throughout the paper, we assume that $\ruleset$ is finite.
Since the ordering of the rules forms a preorder, comparing rules entails three possibilities:
\begin{inparaenum}
\item Strict priority: One rule is strictly more important than the other,
  denoted as $\rbrule_{1} > \rbrule_{2}$, indicating that $\rbrule_{1}$ has priority over
  $\rbrule_{2}$, e.g., collision avoidance is prioritized over the correct use of turn signals;
\item Incomparability:  The rules are incomparable, meaning that neither
  $\rbrule_{1} \preorder \rbrule_{2}$ nor $\rbrule_{2} \preorder \rbrule_{1}$ holds,
  e.g., safety of property vs. safety of animals;
\item Equal rank: The rules are of the same importance,
  denoted as $\rbrule_{1} \sim \rbrule_{2}$,
  e.g., driving within a drivable area vs. respecting lane direction.
\end{inparaenum}
This preorder structure is more expressive than standard weighted or lexicographic formulations.
Specifically, weighted approaches correspond to rules of equal importance (case (3)), while
lexicographic approaches correspond to a total priority ordering (case (1)).
As shown in \cite{Wongpiromsarn:2026:Formal}, the classical rulebook formalism
generalizes both temporal-logic-based planning and optimization-based control in the deterministic
setting.
% We refer the reader to \cite{Wongpiromsarn:2026:Formal} for more details.

It can be proved \cite{Censi:2019:Liability} that a rulebook $\rulebook$ induces a preorder $\rbleq$ on $\realizations$.
Given realizations $x$ and $y$, $x \rbleq y$ can be
interpreted as
$x$ being ``at least as good as'' $y$, i.e., the degree of violation of the rules by $x$ is at most
as much as that of $y$.
Formally, we say that $x \rbleq y$ if for any rule $\rbrule \in \ruleset$
satisfying $\rbrule(x) > \rbrule(y)$, there exists a higher priority rule $\rbrule' > \rbrule$ such that $\rbrule'(x) < \rbrule'(y)$.
We say that $x$ is ``strictly better than'' $y$, denoted by $x \rblt y$,
if $x \rbleq y$ but $y \not\rbleq x$.
An equivalence relation $\sim_{\rulebook}$ can be defined such that $x \sim_{\rulebook} y$
if and only if $x \rbleq y$ and $y \rbleq x$.
It can be proved \cite{Censi:2019:Liability} that $x \sim_{\rulebook} y$ if and only if $\rbrule(x) =
\rbrule(y)$ for all rules $r \in \ruleset$.
% Finally, we say that $x$ and $y$ are \emph{comparable} if at least one of the relations
% $x \rbleq y$ or $y \rbleq x$ holds.

\vspace{-1mm}
\subsection{Probability and Risk Measures}
\vspace{-1mm}
All measurable spaces are assumed to be standard
Borel~\cite{kechrisClassicalDescriptiveSet1995}, i.e., all \(\sigma\)-algebras are the Borel
\(\sigma\)-algebras generated by open sets of a completely metrizable separable topological space.
Let $S_{1}$ and $S_{2}$ be the underlying sets of measurable spaces.
% Let $S_{1}$ and $S_{2}$ be sets equipped with $\sigma$-algebras.
% Given two sets $S_{1}$ and $S_{2}$ that are equipped with respective $\sigma$-algebras $\algebra_{1}$ and $\algebra_{2}$,
We define $\mfunctions(S_{1}, S_{2})$ as the set of all
measurable functions from $S_{1}$ to $S_{2}$.
% An element $f \in \mfunctions(S_{1}, S_{2})$ is a measurable function $f : S_{1} \to S_{2}$.
Let $(\outcomes, \algebra, \probability)$ denote a probability space, where $\outcomes$ is a nonempty
set of outcomes, $\algebra$ is a standard Borel $\sigma$-algebra
% (e.g., Borel $\sigma$-algebra if $\outcomes$ is a topological space)
over $\outcomes$, and
$\probability : \algebra \to [0, 1]$ is a probability measure on $(\outcomes, \algebra)$.
The set $\mfunctions(\outcomes, \reals)$ can be viewed as the set of functions that assign
a cost to each outcome. Each function $f \in \mfunctions(\outcomes, \reals)$ is referred to as
\emph{a random cost variable}.
A \emph{risk measure} is a mapping $\risk : \mfunctions(\outcomes, \reals) \to \reals$, which assigns
a real-valued assessment of risk to a given random cost variable.

For brevity, given random cost variables $f, f' \in \mfunctions(\outcomes, \reals)$
and a binary relation $\star \in \{ \le, <, \ge, >, = \}$,
we write \(\probability(f \star f')\)
to denote \(\probability\big(\{\outcome \in \outcomes \mid f(\outcome) \star f'(\outcome)
\}\big)\).
Similarly, for $d \in \reals$, we write
$\probability(f \star d)$ to denote
$\probability(\{\outcome \in \outcomes \mid f(\outcome) \star d\})$.
A statement is said to hold \emph{almost surely}
if it is true for all outcomes except for a set of probability zero.
In this notation, $f \star f'$ almost surely means $\probability(f \star f') = 1$.

Risk measures are typically required to satisfy a collection of axioms
that ensure coherent and consistent assessments of uncertainty \cite{Artzner:1999:Coherent,Majumdar:2020:How}.
We only require the monotonicity property,
which captures the minimal consistency condition that larger costs
should not be assessed as less risky.

\begin{assumption}[Monotonicity]
  For any $f, f' \in \mfunctions(\outcomes, \reals)$,
  if $f \leq f'$ almost surely, then $\risk(f) \leq \risk(f')$.
  \label{assumption:monotonicity}
\end{assumption}

\begin{figure}
  \begin{tikzpicture}
  \begin{axis}[
    width=0.9\linewidth,
    height=0.4\linewidth,
    xlabel={Random cost variable $f$},
    ylabel={Probability density},
    xmin=0, xmax=12,
    ymin=0, ymax=0.25,
    axis lines=left,
    axis on top,
    samples=300,
    domain=0:10,
    clip=false,
    xtick={0,2,4,6,8,10},
    tick label style={font=\scriptsize},
    label style={font=\scriptsize},
    legend style={font=\scriptsize},
    ]

    % --- PDF ---
    % \addplot[thick]
    % {3.125*((x + 0.64)*(10 - x)/(x+5)^3)};
    % {0.001 + 0.0014*(10-x)^2};

    % --- Parameters ---
    \def\fmax{10}
    \def\Ef{2.965}
    \def\VaR{6.113}
    \def\CVaR{7.326}
    \def\ytop{0.22}

    % --- Shade area up to VaR to indicate alpha ---
    \addplot[
    domain=0:\VaR,
    draw=none,
    fill=blue!10
    ]
    {3.125*((x + 0.64)*(10 - x)/(x+5)^3)}
    \closedcycle;

    % --- Label alpha with an arrow for clarity ---
    \node (alpha_label) at (axis cs: 4.7, 0.15) {\scriptsize $\alpha = 0.9$};
    \draw[->, >=stealth, thin] (alpha_label) -- (axis cs: 4, 0.08);

    % --- PDF Line (drawn after fill so it's on top) ---
    \addplot[thick, black]
    {3.125*((x + 0.64)*(10 - x)/(x+5)^3)};

    % Expected value
    \draw[dashed, black, thick] (axis cs:\Ef,0) -- (axis cs:\Ef,\ytop)
    node[anchor=south, color=black] {\scriptsize $\mathbb{E}[f]$};

    % VaR
    \draw[solid, red!80!black, thick] (axis cs:\VaR,0) -- (axis cs:\VaR,\ytop)
    node[anchor=south] at (axis cs:\VaR-0.3,0.22) {\scriptsize $\text{VaR}_{0.9}$};

    % CVaR
    \draw[solid, blue!80!black, thick] (axis cs:\CVaR,0) -- (axis cs:\CVaR,\ytop)
    node[anchor=south] at (axis cs:\CVaR+0.3,0.22) {\scriptsize $\text{CVaR}_{0.9}$};

    % Worst Case
    \draw[dashed, black, thick] (axis cs:\fmax,0) -- (axis cs:\fmax,\ytop)
    node[anchor=south] {\scriptsize Worst-case};

    % % --- Label alpha ---
    % \node[anchor=north east] at (axis cs:\VaR/2+1.2,0.05)
    % {$\alpha = 0.9$};

    % --- Expected cost ---
    % \addplot[densely dotted, thick]
    % coordinates {(\Ef,0) (\Ef,\ytop)};
    % \node[anchor=north] at (axis cs:\Ef,0.25)
    % {$\mathbb{E}[f]$};

    % % --- VaR ---
    % \addplot[densely dotted, thick]
    % coordinates {(\VaR,0) (\VaR,\ytop)};
    % \node[anchor=north] at (axis cs:\VaR-0.7,0.25)
    % {$\text{VaR}_{0.9}$};

    % % --- CVaR ---
    % \addplot[densely dotted, thick]
    % coordinates {(\CVaR,0) (\CVaR,\ytop)};
    % \node[anchor=north] at (axis cs:\CVaR+0.4,0.25)
    % {$\text{CVaR}_{0.9}$};

    % % --- Worst case ---
    % \addplot[densely dotted, thick]
    % coordinates {(\fmax,0) (\fmax,\ytop)};
    % \node[anchor=north] at (axis cs:\fmax+0.6,0.25)
    % {Worst-case};

  \end{axis}
\end{tikzpicture}
  \vspace{-3mm}
  \caption{Illustrative probability density of a random cost variable $f$,
    highlighting key risk measures for $\alpha = 0.9$.
    \vspace{-5mm}}
  \label{fig:var}
\end{figure}

Many standard risk measures satisfy this property,
including expected cost, worst-case assessment,
Value at Risk (VaR), and Conditional Value at Risk (CVaR).
The VaR at level $\alpha$ for a random cost variable $f$ is defined as
$\VaR_{\alpha}(f) = \min\{ d \mid \probability(f \leq d) \geq \alpha \}$
and represents the $\alpha$-quantile of $f$
(some definitions \cite{Majumdar:2020:How} alternatively use the $(1-\alpha)$-quantile).
The CVaR at level $\alpha$ for $f$ is defined as
$\CVaR_{\alpha}(f) = \inf_{\beta \in \reals}\big\{\beta + \frac{1}{1-\alpha} \expectation[f -
\beta]^{+} \big\}$, where $[f-\beta]^{+} = \max\{f-\beta, 0\}$.
Under mild smoothness assumptions, %on the cumulative distribution function of $f$,
CVaR simplifies to the expected cost conditioned on exceeding VaR:
$\CVaR_{\alpha}(f) = \expectation\big[f | f > \VaR_{\alpha}(f)\big]$,
a formulation commonly used in literature \cite{Pflug:2000:Some}.
As illustrated in Figure~\ref{fig:var}, $\VaR_{\alpha}$ represents
the lower bound of the $(1-\alpha)$ worst-case outcomes.
However, it does not account for the severity of costs beyond this threshold.
In contrast, $\CVaR_{\alpha}$ captures the ``expectation of the tail'', providing a
measure of the average cost incurred during extreme events.
In long-tailed distributions, the gap between VaR and CVaR
highlights the ability of CVaR to capture tail risk that VaR ignores.
Note that VaR satisfies monotonicity but is not, in general,
subadditive and therefore is not a coherent risk measure.
Our results rely only on monotonicity and therefore
apply to both coherent risk measures (e.g., expected cost, worst-case, CVaR)
and non-coherent ones (e.g., VaR).

\section{Uncertain System-Environment Interaction}
Consider a system that interacts with an external environment.
The system state evolves under its control inputs,
whereas the environment state is not directly controllable.
However, the environment may respond to the resulting system trajectory.
% which is induced by the applied control inputs.
For example, when an autonomous vehicle slows down near a crosswalk,
it may increase the likelihood that a pedestrian crosses in front of the vehicle,
compared to a scenario where the vehicle speeds up.

Let $\rtrajectories$ and $\etrajectories$
represent the sets of system and environment trajectories.
A trajectory represents all information about the evolution of the system or environment that is
relevant for evaluation and interaction, and may, for example, take the form of a continuous-time
signal or a discrete-time sequence.
We assume that both $\rtrajectories$ and $\etrajectories$ are equipped with respective $\sigma$-algebras,
$\algebra_{s}$ for $\rtrajectories$ and $\algebra_{e}$ for $\etrajectories$, such that
that the pairs $(\rtrajectories, \algebra_{s})$ and $(\etrajectories, \algebra_{e})$ form standard Borel
spaces.
% , meaning that the sets of trajectories are measurable with respect to the respective
% $\sigma$-algebras.

The interaction between the system and the environment is modeled by a Borel function
$\interaction : \rtrajectories \times \outcomes \to \etrajectories$,
which is defined on the probability space $(\outcomes, \algebra, \probability)$.
For each system trajectory $\rtraj \in \rtrajectories$,
we define a random variable $\interaction_{\rtraj} : \outcomes \to \etrajectories$, which maps outcomes
from the probability space $(\outcomes, \algebra, \probability)$ to environment trajectories as
$\interaction_{\rtraj}(\outcome) = \interaction(\rtraj, \outcome)$.
The set $\outcomes$ can be interpreted as representing all possible environment types
or scenarios that could influence the interaction between the system and its environment.
Each $\outcome \in \outcomes$ corresponds to a specific environment scenario,
such as the level of aggressiveness or attentiveness of other road users.
These scenarios capture the unique characteristics of the environment, and when paired with a
system trajectory, they define the resulting environment trajectory.
%
% $\interaction_{\rtraj}$ is a random variable for every robot
% trajectory $\rtraj \in \rtrajectories$.

\begin{remark}
  The assumption that $(\rtrajectories, \algebra_{s})$ and $(\etrajectories, \algebra_{e})$
  are standard Borel spaces and that $\interaction$ is a Borel function are mild and
  exclude only pathological cases that are unlikely to occur in real-world applications.
  Borel sets include a broad class of sets commonly encountered in control applications,
  such as finite or countable sets and complete separable metric spaces like $\reals^{n}$.
  % Standard Borel spaces cover a broad class of sets encountered in robotics,
  % including finite and countable spaces, as well as separable metric spaces
  % such as Euclidean spaces and configuration spaces commonly encountered in robotic systems.
  These properties ensure broad applicability while maintaining the necessary mathematical structure
  for defining measurable sets and functions.
  Similarly, requiring the interaction $\interaction$ to be a Borel function ensures
  essential measurability properties for mathematical rigor,
  without imposing restrictive or impractical conditions on the system.
\end{remark}

\begin{example}
  Let $S_{s}$ and $S_{e}$ be separable completely metrizable topological spaces that represent
  the state spaces of the system and the environment, respectively.
  Each of these spaces is equipped with the Borel $\sigma$-algebra, which is generated by the
  open sets in $S_{s}$ and $S_{e}$, respectively.
  Consider trajectories of finite length $T \in \naturals$, so that
  $\rtrajectories \subseteq S_{s}^{T}$ and $\etrajectories \subseteq S_{e}^{T}$
  represent the sets of system and environment trajectories.
  Each trajectory is a finite sequence of states:
  a trajectory $\rtraj \in \rtrajectories$ and $\etraj \in \etrajectories$ are of the form
  $\rtraj = s_{s,1}, \ldots, s_{s,T}$ and
  $\etraj = s_{e,1}, \ldots, s_{e,T}$,
  where $s_{s,i} \in S_{s}$, and $s_{e,i} \in S_{e}$, $\forall i$.
  The sets $\rtrajectories$ and $\etrajectories$ are assumed to be Borel and are equipped with the
  induced $\sigma$-algebras.
\end{example}

\section{Risk-Aware Rulebooks}
Consider defining $\realizations = \rtrajectories \times \etrajectories$
as the set of realizations.
A rulebook $\rulebook$ then induces a preorder on
this set, meaning that it evaluates realizations \emph{after the fact}--that is, once both the
system and environment trajectories are fully determined.

On the other hand, if the environment scenario $\outcome \in \outcomes$ is known,
we can define $\realizations = \rtrajectories$.
The rulebook then induces a preorder $\preorder_{\rulebook, \outcome}$ on system trajectories
$\rtrajectories$ such that for any $\rtraj, \rtraj' \in \rtrajectories$,
$\rtraj \preorder_{\rulebook, \outcome} \rtraj'$ if
$(\rtraj, \interaction(\rtraj, \outcome)) \rbleq
(\rtraj', \interaction(\rtraj', \outcome))$, i.e.,
$\rtraj$ is at least as good $\rtraj'$ under environment scenario $\outcome$.

However, during planning, the system must choose its trajectory without knowing the exact environment
scenario $\outcome$.
At best, it may have access to a probability distribution over $\outcomes$.
To address this challenge, we utilize the concept of risk to establish a preorder among system
trajectories, enabling planning under uncertain environments.

\begin{definition}
  \label{def:rarule}
  Let
  $\rbrule : \rtrajectories \times \etrajectories \to \nnreals$
  be a rule, equipped with a risk measure
  $\rrisk : \mfunctions(\outcomes,\reals)\to\reals$
  and a threshold $\rthreshold \in \reals$.
  Assume that $\rbrule$ is measurable.
  Then, for every $\rtraj \in \rtrajectories$,
  the function $\rbrule_{\rtraj} : \outcomes \to \nnreals$ defined as
  $\rbrule_{\rtraj}(\outcome) = \rbrule(\rtraj, \interaction_{\rtraj}(\outcome))$
  is also measurable, i.e., $\rbrule_{\rtraj} \in \mfunctions(\outcomes, \nnreals)$.
  The \emph{risk-aware rule} induced by
  $(\rbrule,\rrisk,\rthreshold)$
  is the function
  \(\rrrule : \rtrajectories \to \nnreals\)
  defined as
  \begin{equation}
    \rrrule(\rtraj) = \max\{\rrisk(\rbrule_{\rtraj})-\rthreshold,0\}.
    \label{eq:rarule}
  \end{equation}
  In general, different rules $\rbrule$ and $\rbrule'$ may have different risk measures
  and thresholds.
  For notational simplicity, however, we do not explicitly indicate these rule-specific choices in
  the notation $\rrrule$, which denotes the risk-aware version of $\rbrule$
  using its associated risk measure and threshold.
\end{definition}

% \begin{definition}
%   \label{def:rarule}
%   A \emph{risk-aware rule} $\rarule$ consists of a rule
%   $\rbrule : \rtrajectories \times \etrajectories \to \nnreals$,
%   an associated risk measure
%   $\rrisk : \mfunctions(\outcomes, \reals) \to \reals$,
%   and a risk threshold $\rthreshold \in \reals$
%   such that for every system trajectory $\rtraj \in \rtrajectories$,
%   the function $\rbrule_{\rtraj} : \outcomes \to \nnreals$
%   defined as $\rbrule_{\rtraj}(\outcome) = \rbrule(\rtraj,
%   \interaction_{\rtraj}(\outcome))$
%   is measurable, i.e.,
%   $\rbrule_{\rtraj} \in \mfunctions(\outcomes, \nnreals)$.
% \end{definition}

Intuitively, each system trajectory $\rtraj \in \rtrajectories$ induces a random variable
$\rbrule_{\rtraj}$ that quantifies the degree of violation of $\rtraj$
with respect to the rule $\rbrule$ under environmental uncertainty.
% Under this definition, each system trajectory $\rtraj \in \rtrajectories$
% is associated with a random variable $\rbrule_{\rtraj}$, where
% $\rbrule_{\rtraj}(\outcome)$ represents the degree of violation of $\rtraj$
% with respect to the rule $\rbrule$ under the environment scenario $\outcome$.
The distribution of $\rbrule_{\rtraj}$ is determined by the probability measure $\probability$,
defined over the probability space $(\outcomes, \algebra, \probability)$.
The risk of $\rtraj$ with respect to the rule $\rbrule$ is defined as
$\rrisk(\rbrule_{\rtraj})$.
Note that $\rrisk$ can be chosen independently for each rule
and can be any risk measure, including expected value, worst-case, VaR, or CVaR.
To simplify the notation, we use $\rrisk(\rtraj)$ throughout the paper to
represent this risk, i.e., $\rrisk(\rtraj) = \rrisk(\rbrule_{\rtraj})$.
% For each risk-aware rule \(\rarule\), we introduce a rule
% \(\rrrule : \rtrajectories \to \nnreals\) given by
% \(\rrrule(\rtraj) = \max\{\rrisk(\rbrule_{\rtraj})-\rthreshold, 0\}\).
The risk-aware rule \(\rrrule\) measures the amount by which this risk exceeds the prescribed
threshold.
This thresholding treats all trajectories below the threshold as equally acceptable under the
rule, preventing over-optimization of higher-priority rules and allowing lower-priority rules to meaningfully
refine comparisons.

With this definition, we can now define the concepts of ``safety'' and ``riskiness''
of system trajectories.

\begin{definition}
  \label{def:safe_rule}
  A system trajectory $\rtraj \in \rtrajectories$ is
  \emph{safe with respect to a risk-aware rule $\rrrule$}
  if \(\rrrule(\rtraj) = 0\), i.e.,
  the associated risk \(\rrisk(\rtraj)\) does not exceed the threshold $\rthreshold$,
  % i.e., $\rrisk(\rtraj) \leq \rthreshold$
\end{definition}

\begin{definition}
  \label{def:riskier_rule}
  A system trajectory $\rtraj \in \rtrajectories$ is
  \emph{strictly less risky with respect to a risk-aware rule $\rrrule$}
  than $\rtraj' \in \rtrajectories$
  if \(\rrrule(\rtraj) < \rrrule(\rtraj') \), i.e., both of the following conditions are satisfied:
  \begin{enumerate}
  \item $\rrisk(\rtraj') > \rthreshold$ (the risk of $\rtraj'$ exceeds the threshold),
    and
  \item $\rrisk(\rtraj') > \rrisk(\rtraj)$ (the risk of $\rtraj'$ is more than
    that of $\rtraj$).
  \end{enumerate}
  In this case, we say that $\rtraj'$ is \emph{strictly riskier with respect to $\rrrule$}
  than $\rtraj$.
\end{definition}

Given a rulebook \(\rulebook = \langle \ruleset, \preorder \rangle\)
on realizations \(\rtrajectories \times \etrajectories\),
suppose each $\rbrule : \rtrajectories \times \etrajectories \to \nnreals$
is equipped with a risk measure $\rrisk$ and a threshold $\rthreshold$.
Each such rule induces a corresponding risk-aware rule
$\rrrule : \rtrajectories \to \nnreals$,
defined on system trajectories $\rtrajectories$.
The collection of these induced rules, together with the inherited preorder,
forms a risk-aware rulebook on system trajectories,
as formalized below.

% Collecting these induced rules yields a new rule set $\raruleset$.
% The preorder on \(\raruleset\) is inherited from \(\preorder\).
% This construction gives rise to a risk-aware rulebook on system trajectories,
% as formalized in the following definition.

\begin{definition}
  \label{def:rarulebook}
  Let \(\rulebook = \langle \ruleset, \preorder \rangle \) be a rulebook on
  \(\rtrajectories \times \etrajectories\).
  For each \(r \in \ruleset\), let \(\rrrule\) denote the induced
  risk-aware rule on \(\rtrajectories\).
  The induced \emph{risk-aware rulebook} on \(\rtrajectories\) is a tuple
  \(\rarulebook = \langle \raruleset, \preorder_{\text{risk}} \rangle\), where
  \(\raruleset = \{\rrrule \ | \ r \in \ruleset\}\) and
  the preorder \(\preorder_{\text{risk}}\) is defined by
  \(\rrrule \preorder_{\text{risk}} \rrrule'\) if and only if
  \(\rbrule \preorder \rbrule '\).
\end{definition}

\section{Safe and Optimal Trajectories}

A risk-aware rulebook $\rarulebook$ induces a relation $\raleq$ on $\rtrajectories$
and defines the safety of system trajectories.

\begin{definition}
  \label{def:raleq}
  Given a risk-aware rulebook \(\rarulebook = \langle \raruleset, \preorder_{\text{risk}}
  \rangle\),
  we say that a system trajectory $\rtraj \in \rtrajectories$ is \emph{no riskier than}
  $\rtraj' \in \rtrajectories$, denoted by $\rtraj \raleq \rtraj'$,
  if for every risk-aware rule $\rrrule \in \raruleset$ such that
  $\rrrule(\rtraj) > \rrrule(\rtraj')$,
  there exists a higher priority rule $\rrrule' >_{\text{risk}} \rrrule$
  such that $\rrrule'(\rtraj') > \rrrule'(\rtraj)$.
\end{definition}

\begin{definition}
  \label{def:riskier}
  A system trajectory $\rtraj \in \rtrajectories$ is \emph{strictly less risky} than
  $\rtraj' \in \trajectories$, or equivalently,
  $\rtraj'$ is \emph{strictly riskier} than $\rtraj$,
  denoted by $\rtraj \ralt \rtraj'$,
  if $\rtraj \raleq \rtraj'$ but $\rtraj' \not\raleq \rtraj$.
\end{definition}

\begin{definition}
  \label{def:safe}
  A system trajectory $\rtraj \in \rtrajectories$ is \emph{safe}
  if \(\rrrule(\rtraj) = 0\) for all risk-aware rules $\rrrule \in \raruleset$.
\end{definition}

Since $\rarulebook$ consists of a well-defined rule set equipped with a preorder,
it forms a rulebook on $\rtrajectories$ as formally stated below.

\begin{proposition}
  $\rarulebook$ is a rulebook on $\rtrajectories$.
  \label{prop:rulebook}
\end{proposition}

Proposition \ref{prop:rulebook} has several important implications.
First, it follows from \cite{Censi:2019:Liability} that
the induced relation $\raleq$ is a preorder on $\rtrajectories$.
Although the relation $\rbleq$ defines a preorder on $\rtrajectories \times \etrajectories$,
the relation $\raleq$ extends this structure to $\rtrajectories$ alone.
This allows system trajectories to be compared and ranked without requiring
knowledge of the specific environment scenario and
provides a formal basis for defining optimal system trajectories.
Additionally, the construction of the risk-aware rules makes it natural for a system designer to
start from a classical rulebook (e.g., one specified by a regulator and used for retrospective
evaluation under fully realized environmental outcomes) and systematically extend it to the
uncertain setting by equipping each rule with an appropriate risk measure and threshold.

Finally, Proposition \ref{prop:rulebook} ensures that existing control synthesis algorithms
developed for rulebooks \cite{Wongpiromsarn:2026:Formal} can be applied directly.
These algorithms rely on structural assumptions on the rules to ensure tractability.
In the risk-aware setting, this limits their applicability to expectation-based
risk measures and to rules that admit an additive decomposition over concatenated trajectories,
i.e., $r((\tau_{1}, \xi_{1}) \circ (\tau_{2}, \xi_{2})) = r(\tau_{1}, \xi_{1}) + r(\tau_{2},
\xi_{2})$, where $\circ$ denotes concatenation applied componentwise to system and environment trajectories.
This structure enables dynamic programming–type decompositions that are central to efficient synthesis.
While these assumptions are restrictive, they are consistent with common practice in reinforcement
learning and optimal control, where additive cost structures and expectation-based objectives
are often imposed for computational tractability.
At the same time, the risk-aware rulebook formalism provides greater expressiveness by allowing structured
relationships between objectives, rather than requiring them to be aggregated into a single reward
or cost function.
More general risk measures introduce additional computational challenges,
although VaR has been extensively studied under chance-constrained formulations, which
can be interpreted as a special case of a risk-aware rulebook with a
two-level structure (a top level encoding safety via chance constraints and a lower level
performance objective).
Overall, Proposition \ref{prop:rulebook} implies that new algorithms do not need to be developed
specifically for risk-aware rulebooks. Instead, it provides a foundation for extending existing
rulebook-based control synthesis methods to broader classes of rules and risk measures, depending
on the structural assumptions imposed.

\begin{definition}
  \label{def:optimal}
  A system trajectory $\rtraj^{*} \in \rtrajectories$ is \emph{optimal}
  if there is no other trajectory $\rtraj \in \rtrajectories$
  such that $\rtraj \ralt \rtraj^{*}$.
\end{definition}

We now present formal results regarding the relationship between safe and optimal
system trajectories.

\begin{proposition}
  \label{prop:safe-and-optimal-relationship}
  The following relationship between safe and optimal trajectories hold.
  \begin{enumerate}
  \item\label{item:safe-is-optimal} Every safe trajectory is optimal.
  \item\label{item:safe-is-cone} If \(\rtraj\) is safe, then any
    \(\rtraj' \raleq \rtraj\) is also safe and optimal.
  \item\label{item:safe-iff-optimal} If a safe system trajectory exists, then a system trajectory is
    safe if and only if it is optimal.
  \end{enumerate}
\end{proposition}
\begin{proof}
  \ref{item:safe-is-optimal})   Consider a safe trajectory $\rtraj$.
  By Definition \ref{def:safe}, \(\rrrule(\rtraj) = 0\) for all \(\rrrule \in \raruleset\).
  Thus, for any trajectory \(\rtraj' \in \rtrajectories\),
  we have \(\rrrule(\rtraj) \le \rrrule(\rtraj')\) for all \(\rrrule \in \raruleset\)
  so by Definition \ref{def:raleq}, \(\rtraj \raleq \rtraj'\).
  Thus, \(\rtraj\) is optimal.

  \ref{item:safe-is-cone}) From Item \ref{item:safe-is-optimal}, $\rtraj$ is optimal.
  Since $\rtraj' \raleq \rtraj$ but $\rtraj$ is optimal,
  it must be the case that $\rtraj \raleq \rtraj'$, which implies that
  \(\rtraj' \sim_{\rarulebook} \rtraj\).
  So,  \(\rrrule(\rtraj') = 0\) for all $\rrrule \in \raruleset$.

  \ref{item:safe-iff-optimal}) Item~\ref{item:safe-is-optimal} shows that safe trajectories are
  optimal.
  So we only need to show that if a safe trajectory $\rtraj$ exists,
  then any optimal trajectory $\rtraj^*$ is safe.
  Suppose, for contradiction, that $\rtraj^*$ is not safe.
  Then, $\rrrule(\rtraj^*) \geq 0$ for all $\rrrule \in \raruleset$
  and there exist a rule $\rrrule^*$ with $\rrrule^*(\rtraj^*) > 0$.
  Since $\rrrule(\rtraj) = 0$ for all $\rrrule \in \raruleset$,
  we can conclude that $\rtraj \raleq \rtraj^*$ but $\rtraj^* \not\raleq \rtraj$,
  which implies that $\rtraj \ralt \rtraj^*$,
  contradicting the optimality of $\rtraj^*$.
\end{proof}

\section{Rational Tradeoff of Optimal Trajectories}
This section shows that optimal trajectories satisfy principled tradeoff properties under the
risk-aware rulebook.
In particular, if a trajectory is optimal, then any strict improvement under one rule must be
offset by a compensating deterioration under another rule that is not lower in priority.
This property provides a concrete notion of explainability in our setting: preferences between
trajectories can be justified in terms of rule-level comparisons that are consistent with the rule
preorder.
Specifically, when one trajectory is selected over another, the justification can be traced to (i)
the rules under which it incurs strictly smaller risk,
(ii) the fact that this rule is not lower in priority than the competing rules,
and (iii) the existence of a set of environment outcomes with positive probability on which the
corresponding rule violation is strictly smaller.
Thus, optimal trajectories are rational in the sense that their selection is supported by
explicit and structurally consistent tradeoffs across competing objectives under uncertainty.

% \subsection{Optimal Trajectories Are Rational}
% \label{sec:optimal_rational}
Consider a rulebook \(\rulebook = \langle \ruleset, \preorder \rangle\)
on \(\rtrajectories \times \etrajectories\),
where each rule $\rbrule \in \ruleset$
is equipped with a risk measure $\rrisk$ and a threshold $\rthreshold$.
Let $\rtraj^{*} \in \rtrajectories$ be an optimal system trajectory with respect
to the induced risk-aware rulebook \(\rarulebook\) on $\rtrajectories$.
We will show that $\rtraj^{*}$ is rational in the sense that
if another trajectory $\rtraj'$ appears less risky with respect to a particular rule,
then $\rtraj^{*}$ must be favored by some other rule that is not lower in priority.
In other words, any apparent disadvantage under one rule is compensated by an advantage under a rule
that is not less important,
and this advantage occurs on a set of environment scenarios of positive probability.
The following proposition makes this tradeoff principle precise.

 \begin{proposition}
  \label{prop:optimal_risk}
  Suppose $\rrisk$ satisfies
  Assumption \ref{assumption:monotonicity} for all $\rbrule \in \ruleset$
  and there exists a trajectory $\rtraj'$ with $\rrrule'(\rtraj') < \rrrule'(\rtraj^{*})$
  for some $\rbrule' \in \ruleset$.
  Then there must exist a rule $\rbrule^{*} \not< \rbrule'$ such that
  \(\probability(\rbrule^{*}_{\rtraj'} > \rbrule^{*}_{\rtraj^{*}}) > 0\).
  % i.e., there exists a nonzero probability of environment scenarios where
  % $\rtraj'$ incurs a larger violation of $\rbrule^{*}$ than $\rtraj^{*}$.
\end{proposition}
\begin{proof}
  Since $\rrrule'(\rtraj') < \rrrule'(\rtraj^{*})$ but $\rtraj' \not\ralt \rtraj^{*}$,
  there must exist a rule $\rbrule^{*} \not< \rbrule'$ such that
  $\rrrule^{*}(\rtraj^{*}) < \rrrule^{*}(\rtraj')$, i.e.,
  \(\max\{\risk_{\rbrule^{*}}(\rtraj^{*}) - \threshold_{\rbrule^{*}},0\} <
  \max\{\risk_{\rbrule^{*}}(\rtraj') - \threshold_{\rbrule^{*}},0\}\).
  Hence, \(\risk_{\rbrule^{*}}(\rtraj^{*}) < \risk_{\rbrule^{*}}(\rtraj')\).
  Suppose, for contradiction, that
  \(
  \rbrule^{*}_{\rtraj'}
  \leq
  \rbrule^{*}_{\rtraj^{*}}
  \)
  almost surely.
  Then, by monotonicity of $\risk_{\rbrule^{*}}$, we get
  \(
  \risk_{\rbrule^{*}}(\rtraj')
  \le
  \risk_{\rbrule^{*}}(\rtraj^{*})
  \),
  a contradiction. So we can conclude that
  \(
  \probability(\rbrule^{*}_{\rtraj'} > \rbrule^{*}_{\rtraj^{*}}) > 0 \).
\end{proof}

The previous result shows that optimality enforces rational tradeoffs at the
level of risk measure values. The next proposition strengthens this insight by
investigating comparisons at the level of individual environment scenarios.
Suppose a candidate trajectory performs strictly better than an optimal trajectory under some rule
for a particular environment scenario.
The proposition shows that such a pointwise advantage cannot constitute a genuine overall improvement
because one of the following must hold:
(i) the advantage occurs only on a null set;
(ii) the optimal trajectory is already safe with respect to that rule,
i.e., its risk does not exceed the prescribed threshold; or
(iii) there exists a compensating advantage under another rule that is not
lower in priority.
Thus, even pointwise disadvantages are justified by a principled tradeoff.
This stronger result can be derived under the assumption that risk measures are strictly monotonic:
\begin{assumption}[Strict Monotonicity]
  \label{assumption:strict_monotonicity}
  For any $f, f' \in \mfunctions(\outcomes, \reals)$, if
  $f \leq f'$ almost surely and
  \(\probability(f < f') > 0\), then $\risk(f) < \risk(f')$.
\end{assumption}

\begin{proposition}
  Suppose $\rrisk$ satisfies Assumption \ref{assumption:strict_monotonicity}
  for all $r \in \ruleset$ and there exists a trajectory $\rtraj'$ with
  $\rbrule'_{\rtraj'}(\omega) <
  \rbrule'_{\rtraj^{*}}(\omega)$ for some $\rbrule' \in \ruleset$ and $\outcome \in \outcomes$.
  Then, one of the following holds:
  \begin{itemize}
  \item \(\probability( \rbrule'_{\rtraj'} < \rbrule'_{\rtraj^{*}}) = 0\);
    % i.e., the improvement occurs only on a null set;
  \item $\risk_{\rbrule'}(\rtraj^{*}) \leq \threshold_{\rbrule'}$, i.e.,
    $\rtraj^{*}$ is safe with respect to $\rrrule'$, or
  \item there exists a rule $\rbrule^{*} \not< \rbrule'$ such that
    \(\probability(\rbrule^{*}_{\rtraj'} > \rbrule^{*}_{\rtraj^{*}}) > 0\), i.e.,
    there is a compensating advantage elsewhere at a nonzero-probability set of environments.
  \end{itemize}
\end{proposition}
\begin{proof}
  Let
  \(A = \{\outcome \in \outcomes \mid \rbrule'_{\rtraj'}(\outcome) < \rbrule'_{\rtraj^{*}}(\outcome)\}\).
  If $\probability(A) = 0$, then the first conclusion holds.
  Suppose now that $\probability(A) > 0$.
  If \(\probability(\rbrule'_{\rtraj'} > \rbrule'_{\rtraj^{*}}) > 0\),
  then the third conclusion holds with $\rbrule^{*} = \rbrule'$.
  Thus, it remains to consider the case
  \(\rbrule'_{\rtraj'} \le \rbrule'_{\rtraj^{*}}\) almost surely.
  Since $\probability(A) > 0$ and $\rbrule'_{\rtraj'} \le \rbrule'_{\rtraj^{*}}$ almost surely,
  the strict monotonicity of $\risk_{\rbrule'}$ implies
  \(\risk_{\rbrule'}(\rtraj') < \risk_{\rbrule'}(\rtraj^{*})\).
  If $\risk_{\rbrule'}(\rtraj^{*}) \le \threshold_{\rbrule'}$,
  then the second conclusion holds.
  Otherwise, $\risk_{\rbrule'}(\rtraj^{*}) > \threshold_{\rbrule'}$,
  so $\rrrule'(\rtraj^{*}) > 0$, which implies \(\rrrule'(\rtraj') < \rrrule'(\rtraj^{*})\).
  Since $\rtraj^{*}$ is optimal with respect to
  $\rarulebook$, it follows from Proposition \ref{prop:optimal_risk} that
  there exists a rule $\rbrule^{*} \not< \rbrule$ such that
  \(
  \probability(\rbrule^{*}_{\rtraj'} > \rbrule^{*}_{\rtraj^{*}}) > 0
  \).
\end{proof}

\begin{figure}
  \centering
  \resizebox{\columnwidth}{!}{%
    \begin{tikzpicture}[scale=0.8,>=Stealth]

  % --------------------------------------------------
  % Styles
  % --------------------------------------------------
  \tikzset{
    lane/.style={line width=1pt},
    dashedlane/.style={line width=0.8pt,dashed},
    av/.style={rectangle,fill=blue!30,draw,minimum width=0.8cm,minimum height=0.4cm},
    ped/.style={circle,fill=red!40,draw,minimum size=3mm},
    traj/.style={->,line width=1.2pt},
    dim/.style={<->,line width=0.9pt},
  }

  % --------------------------------------------------
  % Road and sidewalk (pedestrians on same side)
  % --------------------------------------------------
  % Road
  \draw[lane] (0,0) -- (10,0);
  \draw[lane] (0,2) -- (10,2);
  \draw[dashedlane] (0,1) -- (10,1);

  % Sidewalk (below AV lane)
  \draw[lane] (0,-0.6) -- (10,-0.6);
  \node at (9.5,-0.3) {\small sidewalk};

  % --------------------------------------------------
  % AV initial position
  % --------------------------------------------------
  \node[av] (av0) at (1,0.5) {};
  \node[below] at (av0.south) {\small AV};

  % --------------------------------------------------
  % Pedestrians (adjacent to AV lane)
  % --------------------------------------------------
  \node[ped] (p1) at (6.5,-0.3) {};
  \node[ped] (p2) at (7.0,-0.3) {};
  \node[ped] (p3) at (7.5,-0.3) {};
  \node[below] at (7.0,-0.6) {\small pedestrians};

  % --------------------------------------------------
  % 25 m distance (dimension line)
  % --------------------------------------------------
  % \draw[dim] (1,2.3) -- (7,2.3);
  % \node at (4,2.6) {\small 25 m};

  % --------------------------------------------------
  % AV trajectories (same lane, different speeds)
  % --------------------------------------------------
  % tau_1: maintain speed
  \draw[traj] (av0.east) -- ++(6.5,0)
  node[pos=0.92,right,above] {\small $\rtraj_{1}$};

  % tau_2: comfortable deceleration
  \draw[traj] (av0.east) -- ++(5.5,0)
  node[pos=0.92,right,above] {\small $\rtraj_{2}$};

  % tau_3: hard brake to stop
  \draw[traj] (av0.east) -- ++(4.5,0)
  node[pos=0.92,right,above] {\small $\rtraj_{3}$};

  % tau_4: lateral deviation into other lane
  \draw[traj]
  (av0.east)
  .. controls +(1.5,0) and +(-3.0,0) ..
  ++(6.5,0.7)
  node[pos=0.92,right,above] {\small $\rtraj_{4}$};

  % \draw[traj] (av0.east) to[out=0,in=180] ++(4.5,0.8)
  % node[midway,above] {\small $\rtraj_{4}$};

  % \draw[traj] (av0.east) -- ++(4.5,0.9)
  % node[midway,right] {\small $\rtraj_{4}$};

  % --------------------------------------------------
  % Pedestrian trajectories
  % --------------------------------------------------
  % zeta_1: stay on sidewalk
  \node at (8.0,-0.3) {\small $\etraj_{1}$};

  % zeta_2: enter road
  \draw[traj,dashed,red!70] (7,-0.1) -- (7,0.9)
  node[pos=0.35,above,right] {\small $\etraj_{2}$};

\end{tikzpicture}

%%% Local Variables:
%%% mode: LaTeX
%%% TeX-master: "../main"
%%% End:
  }
  \vspace{-6mm}
  \caption{AV–pedestrian interaction with their candidate trajectories.
    \vspace{-5mm}}
  \label{fig:ex1:setup}
\end{figure}

\section{Example}
\noindent\textbf{System:}
An autonomous vehicle (AV) is navigating a two-lane road, with a group of pedestrians walking along
the sidewalk adjacent to its lane, 25 meters ahead, as shown in Figure \ref{fig:ex1:setup}.
While the pedestrians are on the sidewalk, there is a small but non-zero probability of erratic
behavior caused by inattention (e.g., a child stepping onto the road without looking or a
distracted person on their phone) or an accident (e.g., a person tripping and falling
into the road).%
\footnote{This example is motivated by a real-world incident in Romania in which a Tesla vehicle
  swerved into the opposite lane to avoid a person who had fallen onto the road, resulting in a
  collision with an oncoming car.}
Additionally, some pedestrians intending to cross the road may unexpectedly do so even without a crosswalk.
The likelihood of such intentional jaywalking increases if the AV slows down, as this may be misinterpreted as
yielding the right of way.
% Such low-probability events, though unlikely, pose potential safety risks that the AV has to
% incorporate into its decision-making process.

To capture the diverse behaviors of pedestrians and their tendencies to enter the road,
a broad range of environment scenarios is needed to account for factors such as
the timing of pedestrian entering the road, their aggressiveness, etc.
To enable a detailed walkthrough and clearly demonstrate how the proposed
formalism captures various aspects of the problem, we focus on 4 distinct cases and define
$\outcomes = \{\outcome_{1}, \outcome_{2}, \outcome_{3}, \omega_{4}\}$.
Here, $\outcome_{1}$ represents a nominal scenario where all pedestrians behave rationally and never
enter the road.
$\outcome_{2}$ corresponds to an erratic scenario where a pedestrian enters the road regardless of
the AV's behavior.
$\outcome_{3}$ and $\outcome_{4}$ involve pedestrians who base their decision to cross on the AV's
deceleration: $\outcome_{3}$ corresponds to crossing if the AV decelerates to any degree,
while $\outcome_{4}$ requires the deceleration to exceed 3 m/s$^{2}$.

While the AV can respond to this situation in numerous ways, to illustrate the application of the
proposed formalism, let us focus on four distinct and contrasting AV's trajectories:
\begin{itemize}
\item $\rtraj_{1}$: maintain its current speed at 15 m/s and stays centered in its lane like most typical human drivers.
\item $\rtraj_{2}$: slow down with a comfortable deceleration of 1 m/s$^{2}$ to
  reduce the severity of a potential collision.
\item $\rtraj_{3}$: abruptly brake at 4.5 m/s$^2$ to come to a complete stop, fully eliminating collision risk.
\item $\rtraj_{4}$: maintain its current speed at 15 m/s while deviating 1 m from its lane to ensure
  sufficient lateral clearance from the pedestrians such that even if a pedestrian enters the
  road, a collision will be avoided. %, at the cost of violating the lane-keeping rule.
\end{itemize}

Similarly, consider two contrasting pedestrian's trajectories
$\etraj_{1}$ and $\etraj_{2}$, which corresponds to no pedestrian entering the road and a
pedestrian entering the road, respectively.

Let $\rtrajectories = \{\rtraj_{1}, \rtraj_{2}, \rtraj_{3}, \rtraj_{4}\}$,
$\etrajectories = \{\etraj_{1}, \etraj_{2}\}$, and
$\algebra = 2^{\outcomes}$.
The interaction between the AV and pedestrians is modeled by the function
$\interaction : \rtrajectories \times \outcomes \to \etrajectories$,
which determines the pedestrian trajectory
resulting from a given AV trajectory $\rtraj_{i} \in \rtrajectories$ and
an environment scenario $\outcome_{j} \in \outcomes$.
The function $\interaction$ and the probability measure $\probability$,
which specifies the likelihood of each scenario $\outcome_{j} \in \outcomes$,
are summarized in Figure \ref{fig:ex1:interaction}.

\noindent\textbf{Rulebook:}
We adopt the rulebook proposed in \cite{Helou:2021:Reasonable}.
For this situation, the relevant set of rules is
$\ruleset = \{\rbrule_{1}, \ldots, \rbrule_{4}\}$, where
\begin{itemize}
\item $\rbrule_{1}$: avoiding collisions with pedestrians, using the speed at collision squared as the measure of violation,
\item $\rbrule_{2}$: lane keeping, with the degree of violation measured by the amount of deviation from the lane,
\item $\rbrule_{3}$: helping flow of traffic, using the difference between the speed limit and the AV's
  minimum speed as the measure of violation, and
\item $\rbrule_{4}$: passenger comfort, with the difference between the actual deceleration and the
  comfortable deceleration squared as the measure of violation.
\end{itemize}
The preorder on $\ruleset$ is defined as
$\rbrule_{1} > \rbrule_{2} > \rbrule_{3}, \rbrule_{4}$, while
$\rbrule_{3}$ and $\rbrule_{4}$ are not comparable.
The value of $\rbrule_{k}(\rtraj_{i}, \etraj_{j})$ for each $i, k \in \{1, \ldots, 4\}$ and
$j \in \{1, 2\}$ is shown in Figure \ref{fig:ex1:rules}.

\begin{figure}[t]
  \centering
  \footnotesize
  \begin{tabular}{|c|c|c|c|c|}
    \hline
    Environment scenario& $\outcome_1$ & $\outcome_2$ & $\outcome_3$ & $\outcome_4$ \\ \hline
    $\probability$& 0.98 & 0.001 & 0.009 & 0.01 \\ \hline
    $\interaction_{\rtraj_1}$ & $\etraj_1$ & $\etraj_2$ & $\etraj_1$ & $\etraj_1$ \\ \hline
    $\interaction_{\rtraj_2}$ & $\etraj_1$ & $\etraj_2$ & $\etraj_2$ & $\etraj_1$ \\ \hline
    $\interaction_{\rtraj_3}$ & $\etraj_1$ & $\etraj_2$ & $\etraj_2$ & $\etraj_2$ \\ \hline
    $\interaction_{\rtraj_4}$ & $\etraj_1$ & $\etraj_2$ & $\etraj_1$ & $\etraj_1$ \\ \hline
  \end{tabular}
  \caption{The interaction $\interaction_{\rtraj_{i}}(\outcome_{j})$ and
    probability measure $\probability(\outcome_{j})$ for each $i, j \in \{1, \ldots, 4\}$.}
  \label{fig:ex1:interaction}
\end{figure}

\begin{figure}[tb]
  \centering
  \footnotesize
  \begin{tabular}{|c||c|c||c|c||c|c||c|c|}
    \hline
    &\multicolumn{2}{c||}{$\rbrule_{1}(\rtraj_i, \etraj_j)$}
    &\multicolumn{2}{c||}{$\rbrule_{2}(\rtraj_i, \etraj_j)$}
    &\multicolumn{2}{c||}{$\rbrule_{3}(\rtraj_i, \etraj_j)$}
    &\multicolumn{2}{c|}{$\rbrule_{4}(\rtraj_i, \etraj_j)$} \\ \hline
    & $\etraj_1$ & $\etraj_2$ & $\etraj_1$ & $\etraj_2$ & $\etraj_1$ & $\etraj_2$ & $\etraj_1$ & $\etraj_2$ \\ \hline
    $\rtraj_1$ & 0 & 225 & 0 & 0 & 0    & 0    & 0   & 0     \\ \hline
    $\rtraj_2$ & 0 & 175 & 0 & 0 & 1.77 & 1.77 & 0   & 0     \\ \hline
    $\rtraj_3$ & 0 &   0 & 0 & 0 & 15   & 15   & 12.25 & 12.25 \\ \hline
    $\rtraj_4$ & 0 &   0 & 1 & 1 & 0    & 0    & 0   & 0      \\ \hline
  \end{tabular}
  \caption{The degree of rule violations for each combination of AV and pedestrian trajectories.
    \vspace{-6mm}}
  \label{fig:ex1:rules}
\end{figure}

\begin{figure}[t]
  \begin{minipage}[t]{0.26\textwidth}
    \adjincludegraphics[width=1.0\textwidth,trim={{0.0\width} {0.0\height} {0.1\width}
      {0.11\height}},clip]{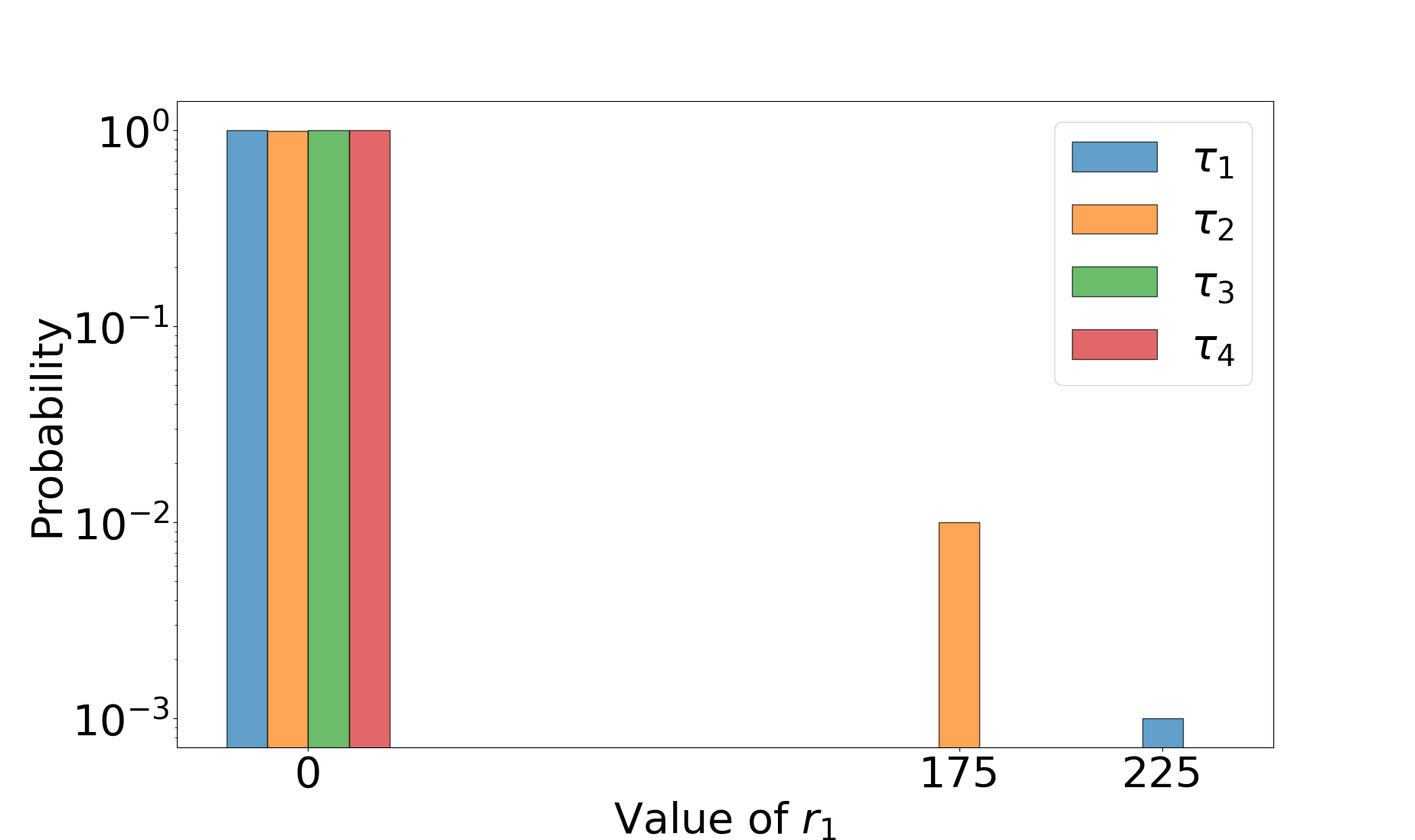}
  \end{minipage}
  \hfill
  \begin{minipage}[t]{0.2\textwidth}
    \vspace{-2.5cm}
    \footnotesize
    \begin{tabular}{|c|c|c|c|}
      \hline
      % &\multicolumn{4}{c|}{$\ruleset(\rtraj_i, \interaction_{\rtraj_{i}}(\outcome_j)$} \\ \hline
        & $\rbrule_{2}$ & $\rbrule_{3}$ & $\rbrule_{4}$ \\ \hline
      $\rtraj_{1}$ & 0 & 0 & 0 \\ \hline
      $\rtraj_{2}$ & 0 & 1.77 & 0 \\ \hline
      $\rtraj_{3}$ & 0 & 15 & 12.25 \\ \hline
      $\rtraj_{4}$ & 1 & 0 & 0 \\ \hline
    \end{tabular}
  \end{minipage}
  \vspace{-2mm}
  \caption{
    The degree of rule violations for different AV trajectories.
    [Left] Probability mass function of rule $\rbrule_{1}$ values for different AV
    trajectory $\rtraj_{i} \in \rtrajectories$ under environmental uncertainty.
    The x-axis represents the realized value $r_{1}(\rtraj_{i})$
    and the y-axis represents the probability of observing that value according to
    the probability measure $\probability$ over the environment scenarios $\outcomes$
    defined in Figure \ref{fig:ex1:interaction}.
    [Right] The violations for $\rbrule_{2}$, $\rbrule_{3}$, and $\rbrule_{4}$,
    where the degree of violation is independent of the environment scenarios. \vspace{-2mm}}
  \label{fig:ex1:violations}
\end{figure}

\begin{figure}[t]
  \footnotesize
  \begin{tabular}{|@{\hskip 1mm}c@{\hskip 1mm}|@{\hskip 1mm}c@{\hskip 1mm}|@{\hskip 1mm}c@{\hskip 1mm}|@{\hskip 1mm}c@{\hskip 1mm}|@{\hskip 1mm}c@{\hskip 1mm}|}
    \hline
    & Expected & Worst case & $\VaR_{\alpha}$ & $\CVaR_{\alpha}$ \\ \hline
    $\rtraj_1$ & 0.225 & 225 & $\begin{array}{ll}
                                  0 &\hspace{-3mm}\text{if } \alpha \leq 0.999\\
                                  225 &\hspace{-3mm}\text{otherwise}
                                \end{array}$
               & $\begin{array}{ll}
                    \frac{0.225}{1-\alpha} &\hspace{-3mm}\text{if } \alpha \leq 0.999\\
                    225 &\hspace{-3mm}\text{otherwise}
                  \end{array}$\\ \hline
    $\rtraj_{2}$ & 1.75 & 175 & $\begin{array}{ll}
                                   0 &\hspace{-3mm}\text{if } \alpha \leq 0.99\\
                                   175 &\hspace{-3mm}\text{otherwise}
                                 \end{array}$
               & $\begin{array}{ll}
                    \frac{1.75}{1-\alpha} &\hspace{-3mm}\text{if } \alpha \leq 0.99\\
                    175 &\hspace{-3mm}\text{otherwise}
                  \end{array}$ \\ \hline
    $\rtraj_{3}$ & 0 & 0 & 0 & 0\\ \hline
    $\rtraj_{4}$ & 0 & 0 & 0 & 0 \\ \hline
  \end{tabular}
  \caption{The risk of each AV trajectory with respect to rule $\rbrule_{1}$,
    evaluated using widely used risk measures.\vspace{-6mm}}
  \label{fig:ex1:risk}
\end{figure}

% \noindent\textbf{Results:}
Figure \ref{fig:ex1:violations} summarizes the probability distribution
of rule violation values $\rbrule_{k}(\rtraj_{i})$ for different AV trajectories
$\rtraj_{i} \in \rtrajectories$ and rules $\rbrule_{k} \in \ruleset$.
This distribution is obtained by combining the value of $\rbrule_{k}(\rtraj_{i}, \etraj_{j})$
from Figure \ref{fig:ex1:rules},
% which quantifies the degree of violations for each pair of AV and pedestrian trajectories,
with the probability distribution of pedestrian trajectories induced by the interaction
$\interaction_{\rtraj_{i}}$, where each $\etraj_{j}$ is
assigned probability $\probability(\{\outcome \in \outcomes \mid
\interaction_{\rtraj_{i}}(\outcome) = \etraj_{j}\})$,
according to the measure $\probability$ in Figure \ref{fig:ex1:interaction}.
% which provides the probability distribution of pedestrian trajectories corresponding to each AV trajectory.

\noindent\textbf{Analysis:}
For rule $\rbrule_{1}$, the choice of risk measure significantly affects the risk
$\rriski{1}(\rtraj_{i})$ of an AV trajectory $\rtraj_{i}$ as summarized in Figure \ref{fig:ex1:risk}.
As a result, the choice of risk measure $\rriski{1}$ and its corresponding risk threshold
$\rthresholdi{1}$ has a significant impact on the ranking of AV trajectories.
For example, if $\VaR_{\alpha}$, with $\alpha \leq 0.999$, is chosen as $\rriski{1}$ with
$\rthresholdi{1} = 0$, then $\rtraj_{1}$ would be the optimal trajectories as it does
not violate any other rules.
This preference can be justified by considering the risk of collision to be acceptable.
Note that this choice of risk measure and threshold is equivalent to imposing
a chance constraint \cite{Nemirovski:2007:Convex,Ono:2008:Efficient} that limits the probability of collision to at most 0.001.

On the other end of the spectrum, a more cautious AV may impose a stricter tolerance for collision.
For example, if $\VaR_{\alpha}$, with $\alpha > 0.999$ (i.e., allowing less than 0.001 probability
of collision), is chosen as $\rriski{1}$ with $\rthresholdi{1} = 0$, then
$\rtraj_{3}$ and $\rtraj_{4}$ would be preferred over $\rtraj_{1}$ and $\rtraj_{2}$ because
the risk of $\rtraj_{1}$ and $\rtraj_{2}$ with respect to the highest priority rule $\rbrule_{1}$
becomes unacceptable, while $\rtraj_{3}$ and $\rtraj_{4}$ only violate lower-priority rules.
The preference between $\rtraj_{3}$ and $\rtraj_{4}$ depends on the acceptable threshold
$\rthresholdi{2}$ for violations of the lane-keeping rule.
If $\rthresholdi{2} < 1$, then $\rtraj_{3}$ will be more preferable.
If $\rthresholdi{2} \geq 1$, then $\rtraj_{4}$ will be more preferable.

Certain combinations of $\rriski{1}$ and $\rthresholdi{1}$ can make $\rtraj_{2}$ the most
preferable option.
These combinations include using worst-case risk measure or
$\VaR_{\alpha}$ with $\alpha > 0.999$ and $175 \leq \rthresholdi{1} < 225$.
Another possibility is using $\CVaR_{\alpha}$ with $\alpha > \frac{174.775}{175}$.
Under these choices, $\rtraj_{2}$ is optimal.
Even though it is less preferable than $\rtraj_{1}$ with respect to $r_{3}$,
Proposition \ref{prop:optimal_risk} explains why this does not contradict optimality.
In particular, there is another rule $r_{1}$ that is not lower in priority than $r_{3}$
under which $\rtraj_{2}$ has an advantage over $\rtraj_{1}$ on a set of environment scenarios with
positive probability.
Specifically, under environment scenario $\omega_{2}$, we have
$r_{1}(\tau_{2}, \interaction_{\tau_{2}}(\omega_{2})) < r_{1}(\tau_{1},
\interaction_{\tau_{1}}(\omega_{2}))$ with $\probability(\omega_{2}) > 0$.
Thus, the disadvantage of $\rtraj_{2}$ under $r_{3}$ is compensated by an advantage under $r_{1}$,
which is at least as important.
This illustrates the tradeoff principle in Proposition~\ref{prop:optimal_risk} and justifies the
rationality of selecting $\rtraj_{2}$.

In summary, there is no universally best AV trajectory, as the most preferable one depends
on the chosen risk measures and the acceptable level of risk.
This aligns with findings in \cite{Helou:2021:Reasonable}, which uses preference data to show that
humans disagree on the best way to drive.
The proposed formalism not only provides an explanation for why a specific trajectory is selected,
but also offers a structured approach for effectively communicating these decisions to regulatory
bodies.

% \begin{example}
%   Consider the previous scenario with an extra complication.
%   Besides the pedestrians on its right, there is an oncoming vehicle is traveling in the opposite
%   lane, posing a direct risk if the AV veers too close to the centerline.
% \end{example}

\vspace{-1mm}
\section{Conclusions and Future Work}
We proposed a risk-aware rulebook formalism for evaluating system trajectories under environmental
uncertainty.
It models system-environment interactions and extends the original rulebook by
equipping each rule with a risk measure and threshold to assess the risk of individual system
trajectories.
We proved that this structure induces a preorder over system trajectories, ensuring
consistency and a well-defined notion of optimality, and showed that optimal trajectories satisfy
principled tradeoff properties across competing rules under uncertainty.
Our formulation generalizes various motion planning formulations,
including those involving hard and soft constraints
(by accommodating multiple levels of constraint hardness),
chance constraints, and risk-aware temporal logic specifications.
Together, our approach provides a foundation for explainable
decision-making in safety-critical autonomous systems.
In practice, we envision this formulation as being particularly useful in
systems where multiple planners (e.g., model-based and learning-based) generate a small set of
candidate trajectories that are then compared using the risk-aware rulebook under uncertainty.
More broadly, this work opens up new research directions in the
development of planning, control, and verification algorithms compatible with risk-aware rulebooks.

\section{Acknowledgments}
The author gratefully acknowledges Konstantin Slutsky for valuable discussions and guidance that
strengthened the mathematical rigor of this work, especially regarding the measure-theoretic and
Borel-space aspects.

\bibliographystyle{ieeetr}
\bibliography{ref}

\end{document}